\documentclass[11pt,reqno]{amsart}
\usepackage{amsmath,amssymb,amsfonts,amsthm,enumerate,natbib,color,ifthen,hyperref}
\usepackage{xargs}
\usepackage[textwidth=4cm, textsize=footnotesize]{todonotes}
\usepackage[latin1]{inputenc}

\usepackage{float}
\usepackage{caption}
\usepackage{subcaption}
\usepackage{graphicx}

\usepackage{aliascnt,bbm}
\def\rset{\mathbb R}
\def\zset{\mathbb Z}

\def\eqsp{\;}

\newcommand{\pscal}[2]{\left\langle#1,#2\right\rangle}

\newcommand{\eqdef}{\ensuremath{\stackrel{\mathrm{def}}{=}}}

\def\B{\mathcal{B}} 
\def\e{\mathcal{E}}

\def\D{\mathcal{D}}

\newcommandx\sequence[3][2=t,3=\zset]
{\ifthenelse{\equal{#3}{}}{\ensuremath{\{ #1_{#2}\}}}{\ensuremath{\{ #1_{#2}, \eqsp #2 \in #3 \}}}}

\def\PP{\mathbb{P}} 
\newcommand{\CPP}[3][]
{\ifthenelse{\equal{#1}{}}{{\mathbb P}\left(\left. #2 \, \right| #3 \right)}{{\mathbb P}_{#1}\left(\left. #2 \, \right | #3 \right)}}
\def\PE{\mathbb{E}} 
\newcommand{\CPE}[3][]
{\ifthenelse{\equal{#1}{}}{{\mathbb E}\left[\left. #2 \, \right| #3 \right]}{{\mathbb E}_{#1}\left[\left. #2 \, \right | #3 \right]}}

\def\tv{\mathrm{tv}}

\usepackage{bm}

\theoremstyle{plain}
\newtheorem{theorem}{Theorem}
\newtheorem{assumption}{H\hspace{-3pt}}

\newaliascnt{proposition}{theorem}
\newtheorem{proposition}[proposition]{Proposition}
\aliascntresetthe{proposition}

\newaliascnt{lemma}{theorem}
\newtheorem{lemma}[lemma]{Lemma}
\aliascntresetthe{lemma}
\newaliascnt{corollary}{theorem}

\aliascntresetthe{corollary}

\theoremstyle{definition}
\newaliascnt{definition}{theorem}

\aliascntresetthe{definition}

\newtheorem{algorithm}{Algorithm}

\newaliascnt{remark}{theorem}
\newtheorem{remark}[remark]{Remark}
\aliascntresetthe{remark}

\newaliascnt{example}{theorem}

\aliascntresetthe{example}

\def\rmd{\mathrm{d}}

\def\1{\mathbbm{1}}

\begin{document}

\title[On cyclical MCMC sampling]{On cyclical MCMC sampling}
\thanks{This work is partially funded by the NSF grant DMS-2210664. We are also grateful to Joonha Park for very helpful comments on an initial draft of this work.}

\author{Liwei Wang}\thanks{ L. Wang: Boston University, 111 Cummington Mall, Boston 02215 MA, United States. {\em E-mail address:} wlwfoo@bu.edu}
\author{Xinru Liu}\thanks{ X. Liu: Boston University, 111 Cummington Mall, Boston 02215 MA, United States. {\em E-mail address:} xinruliu@bu.edu}
\author{Aaron Smith}\thanks{ A. Smith: University of Ottawa, STEM Complex, 150 Louis-Pasteur, Ottawa, Ontario, Canada. K1N 6N5. {\em E-mail address:} asmi28@uOttawa.ca}
\author{Yves Atchad\'e}\thanks{ Y. Atchad\'e: Boston University, 111 Cummington Mall, Boston 02215 MA, United States. {\em E-mail address:} atchade@bu.edu}
  
\subjclass[2010]{62F15, 62Jxx}

\keywords{MCMC, Bayesian deep learning, Tempering, Markov chain convergence rate}

\maketitle

\begin{center} (Feb. 2024) \end{center}

\begin{abstract}
Cyclical MCMC is a novel MCMC framework recently proposed by \cite{zhang2019cyclical} to address the challenge posed by high-dimensional multimodal posterior distributions like those arising in deep learning. The algorithm works by generating a nonhomogeneous Markov chain that tracks -- cyclically in time -- tempered versions of the target distribution. We show in this work that cyclical MCMC converges to the desired probability distribution in settings where the Markov kernels used are fast mixing, and sufficiently long cycles are employed. However in the far more common settings of slow mixing kernels, the algorithm may fail to produce samples from the desired distribution. In particular, in a simple mixture example with unequal variance where powering is known to produce slow mixing kernels, we show  by simulation that cyclical MCMC fails to converge to the desired limit. Finally, we show that cyclical MCMC typically estimates well the local shape of the target distribution around each mode, even when we do not have convergence to the target.
\end{abstract}

\section{INTRODUCTION}
Over the last few decades, statistics and machine learning have become the dominant framework for scientific knowledge discovery and decision making from data. However uncertainty quantification remains one important aspect where further basic advancement is needed. 
In principle Bayesian statistics provides a coherent learning framework  where uncertainty can be rigorously quantified. However, deploying the Bayesian machinery in practice invariably hinges on the ability to handle high-dimensional and often multimodal posterior distributions. Markov chain Monte Carlo (MCMC) is the state-of-the-art for dealing with this problem (\cite{robertetcasella04,douc2018MarkovChains}). Despite several decades of progress in MCMC, sampling from multimodal distributions, particular in a high-dimensional context, remains extremely challenging. The state-of-the-art for dealing with multimodality is the idea of ``tempering" - that is, building a sequence of distributions that bridges the target distribution and some other distribution that is easier to sample. 
Simulated tempering (ST) and parallel tempering (PT) are the two main algorithms  built on this principle (\cite{geyer91pt,marinarietparisi92,hukushima:nemoto:1995,geyer:thompson95}). However to maintain correctness, these algorithms require additional auxiliary variables and Metropolis steps that significantly increase their implementation costs.

In \cite{zhang2019cyclical} the authors proposed ``cyclical stochastic gradient MCMC", a fast implementation of tempering, that dispenses with the costly auxiliary variables required in ST and PT. The algorithm implements tempering without any Metropolis step in the temperature dimension. 
The algorithm operates in cycles, and is designed such that in the initial part of a cycle there is an exploration of the space to find a mode, and in the second part of a cycle, samples are drawn from the mode found. In this paper we use the generic term ``cyclical MCMC" to refer to algorithms built on that principle. The purpose of this work is to analyze these algorithms.

We found that in general cyclical MCMC does not converge to the correct target distribution. For instance, we show by simulation that in a simple Gaussian mixture model with unequal variance, cyclical MCMC does not correctly recover the weights of the mixture. 

Using a novel adaptation of the spectral gap technique to nonhomogeneous Markov chains, we show that cyclical MCMC does converge to the intended target distribution when the Markov kernels used have fast enough mixing, and an appropriate tempering schedule is selected. We also found that even when the weights of the mixture are poorly approximated, cyclical MCMC typically produces a correct approximation of the shape of the distribution around each mode.

The remaining of the paper is organized as follows. We provide some motivating background in Section \ref{sec:background}. The cyclical MCMC algorithm is described in Section \ref{sec:cmcmc}. Our theoretical results are described in Section \ref{sec:theory}, with most proofs collected in Section \ref{sec:proof:thm:1} and in the supplement.

\section{MOTIVATION: BAYESIAN INFERENCE}\label{sec:background}
Suppose that we have data $\D\eqdef\{({\bf x}_i,{\bf y}_i),\;1\leq i\leq n\}$, where ${\bf y}_i\vert {\bf x}_i\sim f_\theta({\bf y}_i\vert {\bf x}_i)$,  for some statistical model $f_\theta$ with (vectorized) parameter $\theta\in\Theta\subseteq\rset^p$. Assuming independence, the log-likelihood writes
\[\ell(\theta;\D) \eqdef \sum_{i=1}^n \ell_i(\theta),\;\;\mbox{ where }\;\; \ell_i(\theta)\eqdef \log f_\theta({\bf y}_i\vert {\bf x}_i).\]
In the Bayesian paradigm we complement the model with a prior distribution $\mu_0$ on $\theta$ that summarizes all available prior knowledge (our inductive bias), such as boundedness, smoothness, sparsity, etc. The resulting posterior distribution of $\theta$ is the probability measure 
\begin{equation}\label{post:Pi}
\Pi(\rmd \theta\vert \D) \propto e^{\ell(\theta;\D)} \mu_0(\rmd\theta).\end{equation}
$\Pi(\cdot\vert\D)$ captures 
the uncertainty in the estimation of $\theta$. Given a new data point for which we observe ${\bf x}$, we can predict corresponding response ${\bf y}$ by sampling from the posterior predictive distribution
\begin{equation}\label{post:pred}
f(\cdot\vert {\bf x})\eqdef \int_{\rset^p} f_{\theta}(\cdot\vert {\bf x})\Pi(\rmd\theta\vert\D).\end{equation}
The uncertainty in a draw from $f$ is the combination of the uncertainty in learning the model, as reflected by the posterior distribution, plus the inherent uncertain of the model itself. Some authors (\cite{hullermeier2021aleatoric}) use the terms epistemic and aleatoric uncertainty respectively to refer to these two sources of uncertainties.

In practice, the posterior integral in (\ref{post:pred}) is rarely tractable, and so it is often replaced by an approximation 
\begin{equation}\label{post:pred:approx}
\hat{f}(\cdot\vert {\bf x}) \eqdef \frac{1}{K}\sum_{k=1}^K f_{\theta^{(k)}}(\cdot\vert {\bf x}),
\end{equation}
where $\{\theta^{(k)},\;1\leq k\leq K\}$ are approximately sampled from $\Pi(\cdot\vert \D)$. Taking $K=1$, and $\theta^{(1)}$ as the estimate obtained by running stochastic gradient descent (SGD) provides a very poor representation of the epistemic uncertainty. Other approaches such as variational approximation (\cite{graves:11,blei:etal:16}) or dropout (\cite{srivastava:14,gal:etal:16}) produce better approximations of $\Pi(\cdot\vert \D)$, but still, are known to misrepresent the epistemic uncertainty. At the other end of the spectrum, traditional MCMC methods produce (asymptotically) correct samples, but are computationally too expensive in many large applications.  There is therefore a pressing need for methods, such as cyclical MCMC, that aim to strike a better balance between cost and accuracy. However a good theoretical understanding of these algorithms is needed.

\section{CYCLICAL MCMC SAMPLING}\label{sec:cmcmc}
Let $\Theta\subseteq\rset^p$ be the state space. Suppose that we are interested in a probability density $\Pi$ (with respect to the Lebesgue measure) of the form
\begin{equation}\label{def:pi}
\Pi(\theta) =\frac{\exp\left(-\e(\theta)\right)}{Z},\;\;\;\theta\in\Theta,
\end{equation}
where $Z$ is the normalizing constant, and $\e:\Theta\to\rset$ some arbitrary measurable function. 
Let $\beta: [0,1]\to [0,1]$ be a continuously differentiable function with $\beta(0)=\beta(1)=1$. We also assume that $\beta$ is decreasing on $[0,1/2]$, and increasing on $[1/2,1]$. We extend $\beta$ into a function $\beta:\;[0,\infty)\to [0,1]$ by period extension (meaning that for all $x\geq 0$, $\beta(x+1) =\beta(x)$). For example, following  (\cite{zhang2019cyclical}) we consider in our simulations the choice\footnote{In the simulations we actually use $\max(\beta(t),0.001)$  to prevent $\beta(t)=0$ which may be problematic when $\Theta$ is unbounded.}
\begin{equation}
\label{beta:eq}
\beta(t) = \frac{1 + \cos(2\pi t^r)}{2},\;\;\;t\geq 0\end{equation}
for some power $r\geq 1$.

Let $L\geq 1$ be a cycle length. For integer $j\geq 1$, and with $\beta_j\eqdef \beta(j/L)$, we define the density
\begin{equation}\label{def:pi:j}
\Pi_j(\theta) =\frac{1}{Z_j} \exp\left(-\beta_j \e(\theta)\right),\;\;\;\theta\in\Theta.\end{equation}
Since $\beta$ is periodic with period $1$ and $\beta(0)=1$, it holds that $\Pi_L=\Pi_{2L} = \cdots \Pi_{kL}=\Pi$, for all $k\geq 1$. Furthermore, as $j$ increases from $1$ to $L/2$, the distribution $\Pi_j(\rmd \theta)$ becomes more diffuse, and its shape is restored back to $\Pi$ as  $j$ increases from $L/2$ to $L$. 

For each $j\geq 1$, let $M_j$ be a Markov kernel on $\Theta$ that can be used to sample from $\Pi_j$. In the theoretical investigation we will assume that $M_j$ has invariant distribution $\Pi_j$. In practice, Markov kernels that do not maintain $\Pi_j$ as  invariant distribution (such as stochastic gradient Langevin dynamics (SGLD)) are used, but we do not analyze these here. Given some initial distribution $\nu^{(0)}$, the idea of cyclical MCMC as developed in (\cite{zhang2019cyclical}) consists in simulating the nonhomogeneous  Markov chain $\{\theta^{(j)},\;j\geq 0\}$, where $\theta^{(0)}\sim \nu^{(0)}$, and for $j\geq 1$,
\begin{equation}\label{markov:prop}
\theta^{(j)} \;\vert\; \{\theta^{(0)},\ldots,\theta^{(j-1)}\}\;\sim M_j(\theta^{(j-1)},\cdot).\end{equation}
The chain is run for $K$ cycles (meaning for $K\times L$ iterations), and we collect the samples obtained at (or around) the end of each cycle to form $\{\theta^{(kL)},\;1\leq k\leq K\}$ as our sample representation of $\Pi$. A synoptic view of the algorithm is given in Algorithm \ref{algo:1}. 

\begin{remark}
The intuition of the method when $\Pi$ is multimodal is that,   as the kernel $M_j$ changes along the cycle and targets $\Pi_j$, by the middle of the cycle (since $\Pi_{1/2}$ is more diffuse), the sampler is able to escape more easily from any current (local) mode. However by the end of the cycle as $\beta_j$ moves closer to $1$, the algorithm  returns back to targeting $\Pi$. Hence, there is an exploration phase with mode discovery, followed by an exploitation phase where samples are drawn from around the selected mode.
\vspace{-0.7cm}
\begin{flushright} $\square$ \end{flushright}
\end{remark}
\medskip

\begin{remark}
The kernel $M_j$ can be constructed using any standard MCMC algorithm that is applicable (Gibbs sampling, Metropolis-Hastings, HMC, etc...). We note however, since the distributions $\Pi_j$ have markedly different covariance structures, for good performance it is important to properly scale the proposal kernels accordingly. In \cite{zhang2019cyclical} the authors mostly focused on SGLD (\cite{sgld,raginsky:etal:17}) and SGHMC (\cite{chen:fox:etal:14,ma:fox:15}).
\vspace{-0.7cm}
\begin{flushright} $\square$ \end{flushright}
\end{remark}
\medskip

\begin{algorithm}\label{algo:1}[Cyclical MCMC]$\hrulefill$\\
Choose the function $\beta:\;[0,1]\to[0,1]$, the cycle length $L$, the number of cycle $K$, the initial distribution $\nu^{(0)}$, and construct the sequence of nonhomogeneous Markov kernel $\{M_j,\;j\geq 1\}$.
\vspace{-0.3cm}

\begin{enumerate}
\item Draw $\theta^{(0)}\sim \nu^{(0)}$.
\item For $j=1,\ldots,K\times L$, draw 
\[\theta^{(j)} \;\vert\; \{\theta^{(0)},\ldots,\theta^{(j-1)}\}\;\sim M_j(\theta^{(j-1)},\cdot).\]
\item Return $\{\theta^{(kL)},\;1\leq k\leq K\}$.
\end{enumerate}
\vspace{-0.5cm}
$\hrulefill$

\end{algorithm}

\subsection{A toy example}\label{sec:mix:1}
Although the idea of cyclical MCMC is intuitively clear and appealing, the algorithm typically does not converge to the correct limit. For instance, consider a simple one-dimensional mixture density
\[\Pi(\theta) = \frac{1}{2}f_1(\theta) + \frac{1}{2}f_2(\theta).\]
where $f_1$ (resp. $f_2$) is the density of $\textbf{N}(5,1)$ (resp. $\textbf{N}(-5,c^2)$) where $c$ is either $c=1$ (equal variance mixture) or $c=0.1$ (unequal variance mixture). We apply Algorithm \ref{algo:1} to sample from $\Pi$ where $M_j$ is taken as a random walk Metropolis with proposal $\textbf{N}(x, 0.25/\beta_j)$. The initial distribution is $\textbf{N}(0,1)$. We use $K=1000$, $L=5000$, and $\beta$ as in (\ref{beta:eq}) with $r=1$. In Figure \ref{fig:ex:1}, the shaded areas on the left side depict the densities estimated from the cyclical MCMC outputs, whereas the solid black lines on the left side represent the true density. In the top-left plot, where $c=1$ (equal variance), the recovery is excellent and it is hard to distinguish the two curves. In the bottom-left where $c=0.1$ (unequal variance), the recovery is poor. In the unequal variance setting cyclical MCMC systematically fails to correctly estimate the weights of the mixture. In this specific example, the estimate of the weight of $f_1$ is $0.497$ in the equal variance setting -- which is excellent, but $0.87$ in the unequal variance setting. 

\begin{figure}
	\caption{Top left (resp. bottom left) is the density $\Pi$ and its estimate as produced by cyclical MCMC when $c=1$ (resp $c=0.1$). Top right (resp. bottom right) shows the powered density $\Pi^{0.001}$ with $c=1$ (resp $c=0.1$).}
	\includegraphics[width = 0.75\textwidth]{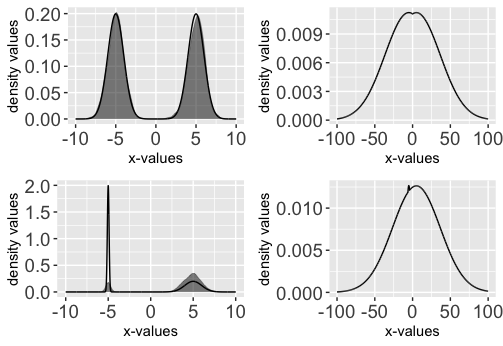}
	\label{alexnet_fashion}
\label{fig:ex:1}
\end{figure}

The issue in this example is that in the unequal variance setting, tempering by powering leads to a witch-hat distribution (\cite{matthew:93}), with a small bulge at the top (bottom-right side of Figure \ref{fig:ex:1}) that is virtually impossible to detect by most MCMC algorithms -- but is crucial to properly approximate the weights of the mixture. As a result the algorithm spends more time on the positive side of the line and produces an estimate  of the mixture weights that is biased towards $f_1$. In the equal variance setting, powering produces a nice smooth symmetric density (top-right side of Figure \ref{fig:ex:1}) and the corresponding Markov kernels $M_j$'s are fast mixing. Parallel tempering and simulated tempering are known to suffer from the same poor mixing in the unequal variance setting (\cite{woodard:09b}). However unlike cyclical MCMC, these algorithms have a built-in Metropolis-Hastings acceptance step that restores the correctness of the target distribution.

\section{SOME THEORETICAL INSIGHTS}\label{sec:theory}
In this section we analyze the convergence of cyclical MCMC. First, we analyze the behavior of Algorithm \ref{algo:1} in the case where the Markov kernels $M_j$'s are fast mixing. In that setting we show that indeed cyclical MCMC can converge to $\Pi$. In the second part we consider a more realistic setting where the cycle of the algorithm can be decomposed into a mode exploration part, and a mode exploitation part. In that regime we show that cyclical MCMC converges to a mixture with correct mixture components, but with possibly incorrect weights.

\subsection{The fast mixing regime}
Since the algorithm generates a nonhomogeneous Markov chain, classical Markov chain mixing time analysis do not apply. We extend the classical Markov chain spectral gap theory to handle nonhomogeneous Markov chains. 

Let $\B_\Theta$ denote the Borel sigma-algebra of $\Theta$. 
We recall that the target density $\Pi$ and the sequence of densities $\Pi_j$ are defined in (\ref{def:pi}) and (\ref{def:pi:j}) respectively. We will abuse notation and also write $\Pi$ (resp. $\Pi_j$) to denote the probability measure on $\Theta$ with density $\Pi$ (resp. $\Pi_j$). For instance we will write $\Pi(A)$ as a short for $\int_A \Pi(\theta)\rmd \theta$. 

To proceed, we need some Markov chain notations. A good reference is (\cite{meyn:tweedie:08}). Given two Markov kernels  $Q_1,Q_2$ on $\Theta$, their product is the Markov kernel $Q_1Q_2$ defined\footnote{In this definition we follow the convenient practice in Markov chains theory of writing the integrant after the integrating measure.} as $(Q_1Q_2)(u,A) \eqdef\int_\Theta Q_1(u,\rmd v)Q_2(v,A)$. This multiplication can naturally be iterated. Given a Markov kernel $Q$, and a probability $\nu$, the product $\nu Q$ denotes the probability measure $(\nu Q)(A) \eqdef \int_\Theta\nu(\rmd u) Q(u,A)$.  Furthermore, given a measure $\nu$ on $\Theta$, and a function $h:\;\Theta\to\rset$, we write $\nu(h)\eqdef \int_\Theta h(u)\nu(\rmd u)$. For $j\geq 1$, we let 
\[\mathcal{L}_j^2\eqdef \mathcal{L}^2(\Pi_j)\eqdef\left\{ h:\;\Theta\to\rset:\;\Pi_j(h^2)<\infty\right\},\]
and for $h,h_1,h_2\in \mathcal{L}_j^2$, we set $\textsf{Var}_j(h) \eqdef \Pi_j(h^2) - \Pi_j(h)^2$ and $\pscal{h_1}{h_2}_j \eqdef \int_\Theta h_1 h_2\rmd\Pi_j$. 
Given two finite measures $\mu,\nu$ on $\Theta$ 
the total variation distance between $\mu$ and $\nu$ is
\[\|\mu -\nu\|_\tv \eqdef \sup_{f:\;|f|\leq 1} \left|\mu(f) - \nu(f)\right|.\]
Given a function $h:\;\Theta\to \rset$, we set $\|h\|_\infty =\sup_{u\in\Theta}\; |h(u)|$.  We impose the following assumption.
\begin{assumption}
\label{H1}
For all $1\leq j\leq L$, the Markov kernel $M_j$ is reversible with respect to $\Pi_j$, and 
\begin{equation}\label{unif:j}
\|\Pi_{j-1} / \Pi_{j}\|_\infty<\infty.
\end{equation}
\end{assumption}

\begin{remark}
Reversibility is imposed here for convenience, and can be removed by introducing the adjoints of the $M_j$'s. Reversibility is a commonly imposed assumption in Markov chain theory, and is satisfied by many MCMC samplers. For instance all Metropolis-Hastings samplers, by definition generate reversible Markov kernels. However there are many other MCMC algorithms that are not reversible. Extending our results beyond the reversible case is an important question for future research.

The boundedness assumption $\|\Pi_{j-1} / \Pi_{j}\|_\infty<\infty$ is typically satisfied when $\Theta$ is bounded, and can essentially be viewed as assuming that $\Theta$ is bounded. Extending our framework to remove that assumption is an important direction for future work.
\vspace{-0.7cm}
\begin{flushright} $\square$ \end{flushright}
\end{remark}

The marginal distribution of $\theta^{(j)}$ at the $j$-th iteration of cyclical MCMC is given by
\[\nu^{(j)} \eqdef \nu^{(0)} M_1\times \cdots \times M_j.\]
We seek conditions under which 
\[\nu^{(kL)} \approx \Pi.\]

In what follows we set
\begin{equation*}
\alpha_j \eqdef \sup\left\{ \Delta_j(f),\;\;f\in \mathcal{L}_{j-1}^2,\;\Pi_{j-1}(f^2)=1\right\},
\end{equation*}
where 
\[\Delta_j(f)\eqdef \int_\Theta\Pi_{j-1}(\rmd x)f^2(x)\int_\Theta M_j^2(x,\rmd z) \left|\frac{\Pi_{j-1}(z)}{\Pi_j(z)}-1\right|.\]
$\alpha_j$ is a measure of discrepancy between $\Pi_{j-1}$ and $\Pi_j$ and can be controlled by the choice of the function $\beta$ and the cycle length $L$. The next result illustrates this point. The proof is given in the supplement.

\begin{proposition}
\label{prop:1}
Suppose that $\Theta$ is bounded, and 
\[\mathsf{osc}(\e)\eqdef\max_{x,y}|\e(x) - \e(y)|,\]
is finite, where $\e$ is as in (\ref{def:pi}). Then there exists a constant $C$, such that for all $1\leq j\leq L$
\begin{equation}\label{bound:alpha}
\alpha_j \leq \frac{C |\dot\beta(t_j)|}{L},\end{equation}
for some $t_j$, with $(j-1)/L \leq t_j \leq j/L$, where $\dot\beta$ is the derivative of $\beta$.
\end{proposition}
\begin{proof}
See Section \ref{sec:proof:prop:1} of the Supplement.
\end{proof}

\begin{remark}
A more natural measure of discrepancy between $\Pi_{j-1}$ and $\Pi_j$ is their total variation distance 
\[\|\Pi_{j-1}- \Pi_j\|_\tv =\int_\Theta \left|\frac{\Pi_{j-1}(z)}{\Pi_j(z)} - 1\right| \Pi_j(z)\rmd z.\]
However the total variation distance does not appear naturally in our analysis.
\vspace{-0.7cm}
\begin{flushright} $\square$ \end{flushright}
\end{remark}

For each $j\geq 1$, let $Q_j:\; \Theta\times \B_\Theta\to\rset$ be the finite kernel  defined as
\begin{equation}\label{def:Qj}
Q_j(x,A) \eqdef \int_\Theta M_j(x,\rmd y)\int_A M_j(y,\rmd z) \frac{\Pi_{j-1}(z)}{\Pi_j(z)}.\end{equation}
Note that for each $x\in\Theta$, $Q_j(x,\cdot)$ is a finite measure on $(\Theta,\B_\Theta)$, but not necessarily a probability measure. We show in Section \ref{sec:proof:prop:1} of the supplement that under H\ref{H1}, $Q_j$ induces an operator $Q_j:\;\mathcal{L}_{j-1}^2\to \mathcal{L}_{j-1}^2$ that is self-adjoint and positive.  We define the ``spectral gap" of the operator $Q_j$ as
\[\lambda_j \eqdef\inf\left\{\mathsf{G}_j(f),\;f\in \mathcal{L}^2_{j-1},\;\textsf{Var}_{j-1}(f)>0\right\},\]
where
\[\mathsf{G}_j(f) \eqdef \frac{\int_\Theta\int_\Theta \left(f(y) - f(x)\right)^2 \Pi_{j-1}(\rmd x)Q_j(x,\rmd y)}{\int_\Theta\int_\Theta \left(f(y) - f(x)\right)^2 \Pi_{j-1}(\rmd x)\Pi_{j-1}(\rmd y)}.\]

\begin{remark}
$\lambda_j$ is large when the spectral gap of $M_j$ is large. This is easily seen from the definition of $Q_j$. One can also easily show for example that if for all $x$, $\|M_j(x,\cdot) - \Pi_j\|_\tv \leq 2(1 - \beta_j)$ for some $\beta_j>0$ that can be viewed as the spectral gap of $M_j$, then $\lambda_j\geq \beta_j$. This is because the last total variation norm inequality is equivalent  to $M_j(x,\cdot) \geq \beta_j \Pi_j(\cdot)$ for all $x$, which in turn implies that $Q_j(x,\cdot) \geq \beta_j\Pi_{j-1}(\cdot)$, and so $\lambda_j\geq \beta_j$.

Finite kernels that are not necessarily Markov appear commonly  in Markov chain theory (for instance in the analysis of sequence Monte Carlo samplers or in large deviations for Markov chains; see e.g. \cite{kontoetal06, herve08, whiteley:13}).
\vspace{-0.7cm}
\begin{flushright} $\square$ \end{flushright}
\end{remark}

With $\alpha_0\eqdef 0$, we set 
\[\Lambda_L \eqdef  \sum_{i=0}^L\alpha_i\prod_{\ell=i+1}^L(1-\lambda_\ell + \alpha_\ell).\]
Under assumption H\ref{H1} we show in (\ref{eq:bound:lambda}) that $\lambda_j \leq 1 + \alpha_j$. Hence it is always true that $\Lambda_L\geq 0$.
Our main result of this section is as follows.
\begin{theorem}\label{thm:1}
Assume H\ref{H1}. Let $\nu^{(0)}(\rmd x) = f_0(x) \Pi(\rmd x)$. 
Then for all $k\geq 1$,
\[
\|\nu^{(kL)}  - \Pi\|_\tv^2 \leq \textsf{Var}_{0}(f_{0})\Lambda_L^k.\]\end{theorem}
\begin{proof}
See Section \ref{sec:proof:thm:1}.
\end{proof}

\begin{remark}
To explore the implications of this result, suppose for instance that 
\begin{equation}\label{eq:fast:mix}
\lambda_j\geq \underline{\lambda}>0,\;\;\mbox{ for all }\;j\geq 1.
\end{equation} Then using Proposition \ref{prop:1}, we see that we can choose the tempering $\beta$ and $L$ such that $\alpha_j \leq C|\dot\beta(t_j)| L^{-1} \leq \underline{\lambda}/2$. In that case we have
\[\left|\Lambda_L\right| \leq \frac{C}{L}\sum_{i=0}^L|\dot\beta(t_i)|\left(1-\frac{\underline{\lambda}}{2}\right)^{L-i}\leq \frac{2C\times \|\dot\beta\|_\infty}{L\underline{\lambda}}<1,\]
for $L> 2C\times \|\dot\beta\|_\infty/\underline{\lambda}$. The condition (\ref{eq:fast:mix}) is a fast mixing condition on the kernels $\{M_j,\;j\geq 1\}$. We note that a similar but more subtle analysis can also be developed in cases where some of the initial Markov kernels $M_j$ (for $j$ closed to $1$) have poor mixing.
\vspace{-0.7cm}
\begin{flushright} $\square$ \end{flushright}
\end{remark}

\begin{remark}
Our theorem and the discussion above thus show that in settings where all the Markov kernels $M_j$ have fast mixing, cyclical MCMC does converge to the right target distribution. For instance, in the mixture density example with equal variance setting, it is well-known that tempering by powering significantly improve mixing (\cite{woodard:etal:09}). Thus we expect cyclical MCMC to work well for $L$ large enough, and this is what we observed in the simulations. In contrast, in the unequal variance setting, it is also known that powering does not improve mixing (\cite{woodard:09b}). Hence in that case all the $\lambda_j$'s remain close to $0$, and cyclical MCMC would require exponentially large cycle length $L$ to work.
\vspace{-0.7cm}
\begin{flushright} $\square$ \end{flushright}
\end{remark}

\begin{remark}
One natural objection to Theorem \ref{thm:1} is this: for many MCMC problems where the tails of $\Pi$ are poorly understood and $\nu^{(0)}$ badly chosen, the variance term $\textsf{Var}_{0}(f_{0})$ appearing in the conclusion of Theorem \ref{thm:1} is infinite, and so the theorem does not give nontrivial upper bounds. Fortunately, in many examples, this problem can be easily fixed via truncation as we explain in Section \ref{sec:rem:thm1:var} of the supplement.
\vspace{-0.7cm}
\begin{flushright} $\square$ \end{flushright}
\end{remark}

\subsection{The highly multimodal regime}
We show here that even when it fails to capture correctly the weights of the mixture, cyclical MCMC typically estimates well the component densities in the mixture. In this section we assume that $\Pi$ is a mixture of the form
\begin{equation}\label{mix:den}
\Pi(\theta) = \sum_{i=1}^d w_i f_i(\theta).\end{equation}
Let $\Theta_{1},\ldots,\Theta_{d} \subset \Theta$ be a collection of disjoint subsets of $\Theta$ such that $\Theta_j$ contains the bulk of the probability mass of $f_j$. For each $j \in [d]$, let $I_{j} \subset \Theta_{j}$, where $[d]$ is a short for $\{1,\ldots,d\}$. Finally, fix $L_{2} \in [L]$. We consider the process $\{\theta^{(j)},\;j\geq 0\}$ generated by Algorithm \ref{algo:1} applied to (\ref{mix:den}). 

The main result in this section, Theorem \ref{ThmMixedModes}, has three main error terms. Immediately after each term is introduced, we verify that it is small for Metropolis-Hastings chains with ultimate target
\begin{equation}\label{EqSimplyGaussMixture}
\Pi = \frac{1}{2}\textbf{N}(-1, \sigma^{2}) + \frac{1}{2} \textbf{N}(-1, c^{2}\sigma^{2})
\end{equation} 
for $0.5 < c < 2$, and proposal kernel  $Unif([\theta-\frac{i}{L} \sigma, \theta+ \frac{i}{L}\sigma])$. These simple chains are similar to the toy example considered in Section \ref{sec:mix:1}. 

Our first error term comes from the following ``no-escape" assumption:

\begin{assumption} \label{AssumptionModeDepth}
There exists some $0 \leq \delta_{1} < 1$ with the following property. For any $j \in [d]$ and any $\theta \in I_{j}$, the nonhomogeneous Markov chain 
\begin{equation}
\theta^{(L_{2})} = \theta, \;\; \theta^{(i)} \sim M_{i}(\theta^{(i-1)},\cdot) \qquad L_{2} < i \leq  L
\end{equation}
satisfies 
\begin{equation}
\PP\left(\bigcup_{i=L_{2}+1}^L \{ \theta^{(i)} \notin \Theta_{j}\}\right) \leq \delta_{1}.
\end{equation}
 
\end{assumption}

\begin{remark}\label{RemLongLyap}
One can verify Assumption \ref{AssumptionModeDepth} by showing that the kernels $M_j$ satisfy a Lyapunov drift condition. We explain the details in Section \ref{sec:checking:AssumptionModeDepth} of the supplement. For the target given in Equation \eqref{EqSimplyGaussMixture}, for $\sigma$ small enough, we take $\Theta_{1} = [-1.5,-0.5]$ and $\Theta_{2} = [0.5,1.5]$ and time $\frac{L}{3} \leq L_{2} \leq \frac{2L}{3}$, and the existence of a Lyapunov drift condition is given by Lemma \ref{LyapunovExplicit} of the supplement.
\vspace{-0.7cm}
\begin{flushright} $\square$ \end{flushright}
\end{remark}

For any $j \in [d]$, denote by $M_{i}^{(j)}$ the Metropolis-Hastings chain with proposal distribution $M_{i}$ and target distribution $\Pi_{i}^{(j)}$ with density proportional to $\Pi_{i}(\theta) \textbf{1}_{\theta \in \Theta_{j}}$; we call this the ``restriction" of $M_{i}$ (respectively $\Pi_{i}$) to the set $\Theta_{j}$. We make the following ``mixing within modes" assumption:

\begin{assumption} \label{AssumptionModeMixing}
There exists $0 \leq \delta_{2} < 1$ with the following property. Fix $j \in [d]$ and $\theta \in I_{j}$. Define the time-inhomogeneous Markov chain:
\begin{equation}
\theta^{(L_{2})} = \theta, \; \theta^{(i)} \sim M_{i}^{(j)}(\theta^{(i-1)},\cdot) \qquad L_{2} < i \leq L.
\end{equation}
This chain satisfies:
\begin{equation}
\| \PP(\theta^{(L)}\in\cdot) - \Pi^{(j)}(\cdot) \|_\tv \leq  \delta_{2}. 
\end{equation}
\end{assumption}

\begin{remark}
The inequality in Assumption \ref{AssumptionModeMixing} is exactly the conclusion of Theorem \ref{thm:1} but applied to the restricted process, and so that theorem can be used to verify Assumption \ref{AssumptionModeMixing}.
\vspace{-0.7cm}
\begin{flushright} $\square$ \end{flushright}
\end{remark}

Under these assumptions, we have:

\begin{theorem}\label{ThmMixedModes}
Let Assumptions \ref{AssumptionModeDepth}, \ref{AssumptionModeMixing} hold and let $\delta = \delta_{1} + \delta_{2}$. Then the Markov chain satisfies:
\[\| \nu^{(L)} - \sum_{j=1}^{d} \nu^{(L_{2})}(I_{j}) \Pi^{(j)} \|_\tv \leq \delta +  (1- \nu^{(L_{2})}(\cup_{j=1}^{d}I_{j})). \]
\end{theorem}
\begin{proof}See Section \ref{sec:proof:ThmMixedModes}.
\end{proof}

\begin{remark}
The remainder term $(1 -  \sum_{j=1}^{d} \nu^{(L_{2})}(I_{j}))$ appearing in Theorem \ref{ThmMixedModes} can be bounded using the same Lyapunov function approach described in Remark \ref{RemLongLyap}. In particular, assume there exists a function $V:\Theta\to [1,\infty)$ and constants $m > 0$, $1 < L_{1} < L_{2}$, $0 < \alpha \leq 1$, and $0 \leq \beta < \infty$ with following properties:
\begin{enumerate}
\item For $\theta\notin \cup_{j} \Theta_j$, we have $V(\theta)\geq e^{m}$.
\item For all $\theta \in \cup_{j} \Theta_{j}$ and all $L_{1} \leq i \leq L_{2}$, one has \[ M_iV(\theta) \leq (1- \alpha) V(\theta) + \beta.\]
\end{enumerate}
Then applying Markov's inequality, we find
\begin{multline*}
\PP\left(X_{L_{2}} \notin \cup_{j} \Theta_j\right) \leq \PP\left(V(X_{L_{2}}) > e^m\right) \\
\leq e^{-m}\PE\left(V(X_{L_{3}})\right) \leq e^{-m}\left(\frac{\beta}{\alpha} + \PE\left(V(X_{L_{1}})\right) \right).
\end{multline*}
Such a Lyapunov function is known to exist in substantial generality. For example, see Theorem 4.1 of \cite{jarnerethansen98} for conditions under which  $V(\theta) = \Pi(\theta)^{-c}$ is a Lyapunov functions for $M_{L}$ for all $0 < c < 1$; this result covers the example in Equation \eqref{EqSimplyGaussMixture}. If we choose $L_{1} \geq a L$, this means that for all $0 < c < \frac{a}{2}$ the single function $V(\theta) = \Pi(\theta)^{-c}$ is simultaneously a Lyapunov function for all chains $M_{i}$ with $L_{1} \leq i \leq L_{2}$. 
\vspace{-0.7cm}
\begin{flushright} $\square$ \end{flushright}
\end{remark}

\begin{remark}
A reviewer insightfully asked whether Theorem \ref{ThmMixedModes} can be  extending to analyze $\nu^{(kL)}$. This can be done using the following general argument. If a Markov kernel $P$ and a probability measure $\nu$ (not nec. invariant under $P$) satisfy $\|P(x,\cdot)-\nu\|_{\mathsf{tv}}\leq \rho$, for all $x$, and for some $\rho\in(0,1)$, then we can say that $P$ satisfies Doeblin's condition, and therefore admits an invariant $\pi$, say, and it holds $\|\pi-\nu\|_{\mathsf{tv}}\leq \rho$. Therefore, for all $k\geq 1$, $\|P^k(x,\cdot)-\nu\|_{\mathsf{tv}}\leq \kappa^k + \rho$, for some $\kappa\in(0,1)$. 

\vspace{-0.7cm}
\begin{flushright} $\square$ \end{flushright}
\end{remark}

\subsubsection{Illustration with a two-dimensional mixture}
As illustration of Theorem \ref{ThmMixedModes}  we consider a two-dimensional Gaussian mixture with 25 components adapted from \cite{zhang2019cyclical}, where
\[\Pi(\theta) = \frac{1}{25} \sum_{i=1}^{25} f_i(\theta),\;\theta\in\rset^2,\]
where $f_i$ is the density of the Gaussian distribution $\textbf{N}(\mu_i,\sigma_i^2 I_2)$. The means $\mu_i$'s  are the elements of $\{-4,-2,0,2,4\}\times \{-4,-2,0,2,4\}$, and $\sigma_i^2=0.2$ in the equal variance setting, and $\sigma_i^2 =0.2/i$, $1\leq i\leq 25$ in the unequal variance setting. In both cases we apply Algorithm \ref{algo:1} with $L=20,000$, $K=50,000$ using a random walk Metropolis with proposal $\textbf{N}(x,0.01\beta_j^{-1/2}I_2)$, and we take $\beta$ as in (\ref{beta:eq}) with $r=1$ (we obtained similar results for values of $r\leq 20$ that we tried). The results are shown on Figure \ref{fig:ex:2}. The estimate of the mixture weights (a uniform distribution on $\{1,\ldots,25\}$) is given on the second row of Figure \ref{fig:ex:2}. Again, we see that in the constant variance setting we recover correctly the weights, but the recovery is poor in the unequal variance setting.

We also look at the $x$-component of $f_{25}$. The shaded area in the top-left (resp. top-right) plot of Figure \ref{fig:ex:2} shows the estimated density in the equal  (resp. the unequal variance) variance setting. The true densities are plotted in dashed-line. We see that in both settings the recovery is good, which confirms the conclusion of Theorem \ref{ThmMixedModes}.

\begin{figure}
	\caption{Top row is the $x$-component of the component density $f_{25}$ and its estimate as produced by cyclical MCMC. The true densities are plotted in blue dashed line. Top-left is the equal variance setting. Bottom left shows the estimated weights of the mixture.}
	\includegraphics[width = 0.75\textwidth]{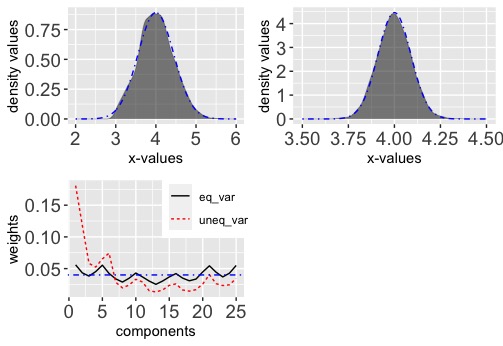}
	\label{alexnet_fashion}
\label{fig:ex:2}
\end{figure}

\section{CONCLUSION REMARKS}
Cyclical MCMC is an efficient implementation of tempering that operates without the costly auxiliary variables required in PT and ST.  The relationship between Cyclical MCMC and ST/PT is similar in some sense to the relationship between the Unadjusted Langevin algorithm and the Metropolis-Adjusted Langevin algorithm \cite{durmus:moulines:19}. 

An initial analysis of the Cyclical MCMC was undertaken by \cite{zhang2019cyclical}. However their convergence upper bounds are actually non-informative. We show in this work that the algorithm does converge to the desired limit in settings where the Markov kernels used are fast mixing, and sufficiently long cycles are used. However in the far more common settings of slow mixing kernels, the algorithm may fail to converge to the correct limit. Indeed, in a simple mixture example with unequal variance where tempering is known to produce slow mixing kernels (\cite{woodard:09b}), we show  by simulation that cyclical MCMC fails to converge to the desired limit. However on the bright side, we also found that even when it fails to capture correctly the weights of the modes, cyclical MCMC typically estimates well the local shape of each mode.

Since it is biased toward flatter modes, in the context of Bayesian inference (Section \ref{sec:background}), it appears that cyclical MCMC leads to a systematic over-estimation of the epistemic uncertainty, which may be a desirable feature in a prediction setting. It has also been observed empirically that in deep learning, flatter modes seem generalize better (see e.g. \cite{pittorino:etal:2021}). It is thus possible that the biased asymptotic behavior of cyclical MCMC that we identified here may become a useful feature in some settings, although more research is needed on this issue.

The bias of the algorithm that we identified in the unequal variance mixture examples is ultimately a limitation of powering as a way of tempering. Our work thus also raises the question of how to build better tempering paths -- that go beyond simple powering -- to obtain fast mixing kernels. There are some recent developments on this issue in the context of PT/ST (\cite{syed:etal:21}). How to leverage these ideas while retaining the initial computational efficiency of cyclical MCMC is also an important direction for future research.


\section{Proof of Theorem \ref{thm:1}}\label{sec:proof:thm:1}
For $j\geq 1$, let $\nu^{(j)}$ denote the marginal distribution of $\theta^{(j)}$, and define the kernel $\bar M_j$ by
\[\bar M_j(u,A) \eqdef \int_A M_{j}(u,\rmd v)\frac{\Pi_{j-1}(v)}{\Pi_{j}(v)}.\]
First, we observe that if $\nu^{(j)}$ admits a density with respect to $\Pi_j$, and $\rmd \nu^{(j)} / \rmd \Pi_j = f_j$,  then $\nu^{(j+1)}$ admits a density with respect to $\Pi_{j+1}$, and
\[\frac{\rmd \nu^{(j+1)}}{\rmd \Pi_{j+1}}(u) = \bar M_{j+1}f_j(u),\;\;\;u\in\Theta.\]
To see this, use reversibility and write for any $A\in\B_\Theta$,
\begin{multline*}
\nu^{(j+1)}(A) = \int_\Theta \Pi_j(\rmd u) f_j(u)\int_\Theta M_{j+1}(u,\rmd v)\textbf{1}_A(v)\\
=\int_A \Pi_{j+1}(\rmd u) \int_\Theta M_{j+1}(u,\rmd v)f_j(v)\frac{\Pi_j(v)}{\Pi_{j+1}(v)}. 
\end{multline*}
Since, by assumption the initial distribution of the cyclical MCMC sampler has a density denoted $f_0$ with respect to $\Pi$, and $\nu^{(0)}(\rmd x) = f_0 (x)\Pi_0(\rmd x)$, we conclude that for all $j\geq 1$, $\nu^{(j)}$ has a density $f_j$ with respect to $\Pi_j$, and the sequence $\{f_j,\;j\geq 0\}$ satisfies 
 \begin{equation}
 \label{eq:rec:fj}
 f_{j+1} = \bar M_{j+1}f_j.
 \end{equation}
 Using this and the Cauchy-Schwarz inequality, we write for all $j\geq 0$,
\begin{equation}\label{eq:1}
\|\nu^{(j)}  - \Pi_{j}\|_\tv =\int_\Theta\left|\frac{\rmd\nu^{(j)}}{\rmd \Pi_j}(u) -1\right| \Pi_j(\rmd u) \leq \sqrt{\textsf{Var}_j\left(f_j\right)}.\end{equation}

We show in Lemma \ref{lem:1} in the supplement that if $f\in\mathcal{L}_{j-1}^2$, then $\bar M_j f\in\mathcal{L}_j^2$, and the adjoint of the operator $\bar M_j$ is $M_j$. Also, recall that $f_j$ is the density of $\nu^{(j)}$ with respect to $\Pi_j$, and we have seen above that $f_j =\bar M_j f_{j-1}$, and $\Pi_j(\bar M_j f_{j-1}) = \Pi_{j-1}(f_{j-1})=1$. Therefore,
\begin{multline*}
\textsf{Var}_j(f_j) =\int_\Theta (\bar M_j f_{j-1})^2 \Pi_j(\rmd x) - 1 = \pscal{\bar M_jf_{j-1}}{\bar M_jf_{j-1}}_{j} - 1 \\
=\pscal{f_{j-1}}{M_j \bar M_jf_{j-1}}_{j-1} - 1 =\pscal{f_{j-1}}{Q_jf_{j-1}}_{j-1} -1,
\end{multline*}
where the operator $Q_j$ is as introduced in (\ref{def:Qj}). Whereas
\[\textsf{Var}_{j-1}(f_{j-1}) = \pscal{f_{j-1}}{f_{j-1}}_{j-1} -1.\]
Hence
\begin{equation}\label{proof:thm:step1}
\textsf{Var}_j(f_j)  =  \textsf{Var}_{j-1}(f_{j-1}) -\pscal{f_{j-1}}{\left( \mathbb{I} - Q_j\right)f_{j-1}}_{j-1},\end{equation}
where $\mathbb{I}$ denotes the identity operator. We now relate the term $\pscal{f_{j-1}}{\left( \mathbb{I} - Q_j\right)f_{j-1}}_{j-1}$ to the spectral gap $\lambda_j$ of $Q_j$. For $f\in \mathcal{L}_{j-1}^2$,
\begin{multline*}
\int_\Theta\int_\Theta \left(f(y) - f(x)\right)^2 \Pi_{j-1}(\rmd x)Q_j(x,\rmd y) \\
= \int_\Theta\Pi_{j-1}(\rmd x) Q_j f^2(x) + \int_\Theta\Pi_{j-1}(\rmd x) f^2(x) Q_j(x,\Theta) \\
-2 \pscal{f}{Q_jf}_{j-1}.
\end{multline*}
We have
\[ Q_j(x,\Theta) = 1 + \int_\Theta M_j(x,\rmd y)\int_\Theta M_j(y,\rmd z) \left(\frac{\Pi_{j-1}(z)}{\Pi_j(z)}-1\right).\]
Hence
\begin{multline*}
\int_\Theta\Pi_{j-1}(\rmd x) f^2(x) Q_j(x,\Theta) = \Pi_{j-1}(f^2)
+ \int_\Theta\Pi_{j-1}(\rmd x)f^2(x)\\
\times \int_\Theta M_j(x,\rmd y)\int_\Theta M_j(y,\rmd z) \left(\frac{\Pi_{j-1}(z)}{\Pi_j(z)}-1\right).\end{multline*}
And since $\Pi_j(Q_jf) = \Pi_{j-1}(f)$, we have
\begin{multline*}
\int_\Theta\Pi_{j-1}(\rmd x) Q_j f^2(x) = \Pi_{j-1}(f^2) + \int_\Theta\Pi_{j}(\rmd x) \left(\frac{\Pi_{j-1}(x)}{\Pi_j(x)}-1\right) Q_j f^2(x)\\
= \Pi_{j-1}(f^2) + \int_\Theta\Pi_{j-1}(\rmd x)f^2(x)\\
\times\int_\Theta M_j(x,\rmd y)\int_\Theta M_j(y,\rmd z) \left(\frac{\Pi_{j-1}(z)}{\Pi_j(z)}-1\right).
\end{multline*}
We conclude that
\begin{multline}\label{eq:qf}
\frac{1}{2}\int_\Theta\int_\Theta \left(f(y) - f(x)\right)^2 \Pi_{j-1}(\rmd x)Q_j(x,\rmd y) = \pscal{f}{\left( \mathbb{I} - Q_j\right)f}_{j-1}  \\
+ \int_\Theta\Pi_{j-1}(\rmd x)f^2(x)\\
\times \int_\Theta M_j(x,\rmd y)\int_\Theta M_j(y,\rmd z) \left(\frac{\Pi_{j-1}(z)}{\Pi_j(z)}-1\right).
\end{multline}
Since $Q_j$ is positive as shown in Lemma \ref{lem:1}, with $\bar f= f - \Pi_{j-1}(f)$, we have $\pscal{\bar f}{\left( \mathbb{I} - Q_j\right)\bar f}_{j-1} \leq \textsf{Var}_{j-1}(f)$, so that
\[
\frac{1}{2}\int_\Theta\int_\Theta \left(f(y) - f(x)\right)^2 \Pi_{j-1}(\rmd x)Q_j(x,\rmd y) \leq \left(1 +\alpha_j\right) \textsf{Var}_{j-1}(f),\]
which implies
\begin{equation}\label{eq:bound:lambda}
\lambda_j \leq 1 + \alpha_j.
\end{equation}
Applied to $f_{j-1}$, (\ref{eq:qf}) gives 
\begin{multline*}
-\pscal{f_{j-1}}{\left( \mathbb{I} - Q_j\right)f_{j-1}}_{j-1} \leq  - \frac{1}{2}\int_\Theta\int_\Theta \left(f_{j-1}(y) - f_{j-1}(x)\right)^2 \Pi_{j-1}(\rmd x)Q_j(x,\rmd y) \\
+\alpha_j \int_\Theta\Pi_{j-1}(\rmd x)f_{j-1}^2(x)\\
\leq -\lambda_j\textsf{Var}_{j-1}(f_{j-1}) + \alpha_j\textsf{Var}_{j-1}(f_{j-1}) + \alpha_j.
\end{multline*}
Taking this last display in (\ref{proof:thm:step1}), we deduce that
\[\textsf{Var}_j(f_j)  \leq \left(1 - \lambda_j + \alpha_j\right) \textsf{Var}_{j-1}(f_{j-1}) + \alpha_j.\]
Iterating this inequality yields
\[
\textsf{Var}_j(f_j) \leq \textsf{Var}_0(f_0)\prod_{k=1}^j(1-\lambda_k+\alpha_k) + \sum_{i=1}^j\alpha_i\prod_{\ell=i+1}^j(1-\lambda_\ell + \alpha_\ell).\]
The result follows from the last display and (\ref{eq:1}).

\section{Proof of Theorem \ref{ThmMixedModes}}\label{sec:proof:ThmMixedModes}
Fix a measurable set $A$ and starting point $\theta$. Since $\nu^{(L)} = \nu^{(L_2)}\prod_{i=L_2+1}^LM_i$,  we have
\begin{multline*}
\nu^{(L)}(A) \geq \sum_{j=1}^{d} \nu^{(L_{2})}(I_{j}) \inf_{\theta' \in I_{j}} (\prod_{i=L_{2}+1}^{L} M_{i,n})(\theta',A)  \geq \sum_{j=1}^{d} \nu^{(L_{2})}(I_{j}) (\Pi^{(j)}(A)-\delta),
\end{multline*}
using Assumptions \ref{AssumptionModeDepth} and \ref{AssumptionModeMixing}.  Similarly,
\begin{multline*}
\nu^{(L)}(A) \leq \sum_{j=1}^{d} \nu^{(L_{2})}(I_{j}) \sup_{\theta' \in I_{j}} (\prod_{i=L_{2}}^{L} M_{i,n})(\theta',A) + (1 -  \sum_{j=1}^{d} \nu^{(L_{2})}(I_{j}))\\
 \leq \sum_{j=1}^{d} \nu^{(L_{2})}(I_{j}) (\Pi^{(j)}(A)+\delta) 
 + (1 -  \sum_{j=1}^{d} \nu^{(L_{2})}(I_{j})).
\end{multline*}
The proof follows from combining these two inequalities.

\bigskip

\newpage
\bibliographystyle{ims}
\bibliography{biblio_mcmc, biblio_new}

\appendix

\section{PROOF OF PROPOSITION \ref{prop:1}}\label{sec:proof:prop:1}
In what follows we write $\beta_j$ as a short for $\beta(j/L)$. We have
\[\frac{\Pi_{j-1}(z)}{\Pi_j(z)} = \exp\left((\beta_j - \beta_{j-1})\e(z) - \log(Z_{j-1}/Z_j)\right).\]
For $t\in [0,1]$, let $\Pi_{j,t}(\theta)\propto \exp(-(t\beta_j+(1-t)\beta_{j-1})\e(\theta))$ be a probability measure that interpolates between $\Pi_{j-1}$ and $\Pi_j$. 
By the path sampling identity (\cite{gelmanetmeng98}),
\[\log\left(\frac{Z_{j-1}}{Z_j}\right)=-\left(\beta_j-\beta_{j-1}\right)\int_0^1\int_\Theta\e(u)\Pi_{j,t}(u)\rmd u.\]
We deduce that
\[(\beta_j - \beta_{j-1})\e(z) - \log\left(\frac{Z_{j-1}}{Z_j}\right)
= \left(\beta_j-\beta_{j-1}\right)\int_0^1\int_\Theta(\e(u) - \e(z))\Pi_{j,t}(u)\rmd u.\]
Since $|e^x-1|\leq xe^{|x|}$, we obtain
\[\left|\frac{\Pi_{j-1}(z)}{\Pi_j(z)} -1\right| \leq |\beta_j - \beta_{j-1}|\mathsf{osc}(\e) e^{|\beta_j - \beta_{j-1}|\mathsf{osc}(\e)}.\]
We take $C = \mathsf{osc}(\e)\times \max_{1\leq j\leq L}e^{|\beta_j - \beta_{j-1}|\mathsf{osc}(\e)}$, and the result follows by first order Taylor expansion of $\beta$.

\section{ON THE TERM $\textsf{Var}_0(f_0)$ IN THEOREM \ref{thm:1}}\label{sec:rem:thm1:var}
One natural objection to the conclusion of Theorem \ref{thm:1} is: for typical MCMC algorithms, the variance term $\textsf{Var}_{0}(f_{0})$ appearing in the conclusion of Theorem \ref{thm:1} is infinite, and so the theorem does not give nontrivial upper bounds. Fortunately, in many examples, this problem can be easily fixed via truncation. More precisely, assume there exist constants $0 < \omega < 1$, $0 < C < \infty$ and $T \in \mathbb{N}$ so that $\prod_{t=0}^{T} M_{t}$ can be written in the form
\[ 
\prod_{t=0}^{T} M_{t}(\theta,\cdot) = \omega H(\cdot) + (1-\omega) R_T(\theta,\cdot),
\] 
where $H$ has a density  $h$ satisfying $\textsf{Var}_{T}(h) \leq C$. In this case, one has
\begin{equation}
\| \prod_{s=0}^{L}M_{s}(\theta,\cdot) - \Pi \|_\tv \leq \| H \prod_{s=T+1}^{L}M_{s} - \Pi \|_\tv + 2(1-\omega),
\end{equation} 
and so Theorem \ref{thm:1} can be applied to the first term starting from time $T$, with the variance term bounded by $C$. 

\section{ON THE KERNELS $Q_j$'s}
\begin{lemma}\label{lem:1}
Assume H\ref{H1}. Then for $f,h\in \mathcal{L}_{j-1}^2$, and $g\in \mathcal{L}_{j}^2$, $\bar M_jf\in \mathcal{L}_j^2$, $\pscal{\bar M_jf}{g}_j = \pscal{f}{M_j g}_{j-1}$. Furthermore, $\pscal{f}{Q_jh}_{j-1} = \pscal{h}{Q_j f}_{j-1}$, and $\pscal{f}{Q_jf}_{j-1} \geq 0$.
\end{lemma}
\begin{proof}
Set $\rho_j\eqdef \|\Pi_{j-1} /\Pi_j\|_\infty$. For $f\in \mathcal{L}_{j-1}^2$, the fact that $\bar M_j f$ belongs to $\mathcal{L}_j^2$ follows from the Cauchy-Schwarz inequality and $\Pi_j M_j=\Pi_j$. Indeed,
\begin{multline*}
\int_\Theta\Pi_j(\rmd u) |\bar M_j f(u)|^2 
\leq \int_\Theta \Pi_j(\rmd u) \int_\Theta M_j(u,\rmd v) f^2(v) \left|\frac{\Pi_{j-1}(v)}{\Pi_j(v)}\right|^2 \\
=  \int_\Theta f^2(v)\frac{\Pi_{j-1}(v)}{\Pi_j(v)} \Pi_{j-1}(\rmd v) \leq \rho_j\Pi_{j-1}(f^2).
\end{multline*}
The equality $\pscal{\bar M_j f}{g}_j = \pscal{f}{M_j g}_{j-1}$ follows easily from reversibility. Indeed, since $\bar M_j f\in \mathcal{L}_j^2$, we have
\begin{multline*}
\pscal{\bar M_j f}{g}_j = \int_\Theta \Pi_j(\rmd u) \int_\Theta M_j(u,\rmd v) g(u) f(v) \frac{\Pi_{j-1}(v)}{\Pi_j(v)} \\
= \int_\Theta \Pi_j(\rmd u) \int_\Theta M_j(u,\rmd v) g(v) f(u) \frac{\Pi_{j-1}(u)}{\Pi_j(u)} = \pscal{f}{M_j g}_{j-1}.
\end{multline*}
A similar argument as above shows that for $f\in\mathcal{L}_{j-1}^2$,
\[\int_\Theta \left(Q_j f(x)\right)^2\Pi_{j-1}(\rmd x) \leq \rho_j^2\Pi_{j-1}(f^2),\]
and
\[ \pscal{f}{Q_jh}_{j-1} = \int_\Theta \Pi_j(\rmd u) \int_\Theta M_j^2(u,\rmd v) f(u) h(v) \frac{\Pi_{j-1}(u)}{\Pi_j(u)} \frac{\Pi_{j-1}(v)}{\Pi_j(v)}  = \pscal{h}{Q_jf}_{j-1}.\]
The positivity is easily seen by noting that $Q_j$ is the product of $M_j$ and its adjoint $\bar M_j$.
\end{proof}

\section{CHECKING ASSUMPTION H\ref{AssumptionModeDepth}}\label{sec:checking:AssumptionModeDepth}
We first check that a Lyapunov drift condition implies Assumption \ref{AssumptionModeDepth}. Specifically, suppose that there exists a function $V:\Theta_j\to [1,\infty)$ and constants $m > 0$, $0 < \alpha \leq 1$, and $0 \leq \beta < \infty$ with following properties:
\begin{enumerate}
\item For $\theta\notin \Theta_j$, we have $V(\theta)\geq e^{m}$.
\item For all $\theta \in \Theta_{j}$ and all $L_{2} \leq i \leq L$, one has \[ M_iV(\theta) \leq (1- \alpha) V(\theta) + \beta.\]
\end{enumerate}
Then we can calculate that for all $\theta \in I_{1}$ and all $L_{2} \leq i \leq L$ we have
\[\PE_\theta\left(V(\theta^{(i)})\right) \leq \frac{\beta}{\alpha}.\]
Thus, by Markov's inequality,
\[\PP\left(X_{i} \notin \Theta_j\right) \leq \PP\left(V(X_{i}) > e^m\right) \leq  e^{-m}\PE\left(V(X_{i})\right) \leq \frac{\beta e^{-m}}{\alpha}.\]
Taking a union bound over $L_{2} \leq i \leq L$, we find that Assumption \ref{AssumptionModeDepth} holds with exponentially small constant $\delta_{1}$ under some mild assumption on the level $m$ of the drift function $V$.

We next note that this proof approach can be used to verify Assumption \ref{AssumptionModeDepth} for targets given in Equation \eqref{EqSimplyGaussMixture} of the main document for sufficiently small $\sigma$. Toward that we take sets $\Theta_{1} = [-1.5,-0.5]$ and $\Theta_{2} = [0.5,1.5]$ and time $\frac{L}{3} \leq L_{2} \leq \frac{2L}{3}$. For $\sigma, s , c > 0$, we consider the Metropolis-Hastings chain $K = K_{\sigma,s,c}$ with target $\mathcal{N}(0, s c^{2} \sigma^{2})$  and proposal $Q(\theta,\cdot) = Unif([\theta-s \sigma, \theta+ s\sigma])$. Denote by $a_{n}$ the acceptance function associated with this Metropolis-Hastings kernel. Finally, for $\alpha > 0$ we consider candidate Lyapunov functions of the form $V(\theta) = e^{\frac{\alpha}{\sigma}|\theta|}$. 

\begin{lemma} \label{LyapunovExplicit}
There exists $\epsilon > 0$ so that for all $s \in [-0.5,2]$ and all $0 < \sigma < \alpha < \epsilon$ sufficiently small, $V$ is a Lyapunov function satisfying:
\[
(KV)(\theta) \leq 0.7 V(\theta) + e^{2 s \alpha} \textbf{1}_{|\theta| \leq  s \sigma}.
\]
\end{lemma}

\begin{remark}
Although we state the result for \textbf{exactly} Gaussian targets on the \textbf{full} real line, it applies with no substantial changes in the following situations:

\begin{enumerate}
\item \textbf{Restricted State Space:} If we change the target from the Gaussian $\mathcal{N}(0, s c^{2} \sigma^{2})$ to the restriction of this Gaussian to an interval $[A,B]$, the result applies with no changes to the constants as long as $A < - s \sigma$ and $B > s \sigma$. To verify this, note that rejecting proposals from within $[A,B]$ to points outside of $[A,B]$ can only decrease the value of $V$.
\item \textbf{Small Multiplicative Perturbations:} Denote by $f$ the density of $\mathcal{N}(0, s c^{2} \sigma^{2})$. If we target instead a distribution with unimodal density $g$ satisfying $\left| \frac{g(\theta)}{f(\theta)} - 1 \right| \leq \epsilon V(\theta)$, then the same Lyapunov inequality holds with the constant $0.7$ replaced by $(0.7 + \epsilon)$. 
\end{enumerate}

These observations can be combined to verify Lyapunov conditions for Gaussian mixture models restricted to regions where one mixture component dominates the density.
\end{remark}

\begin{proof}
Let $\theta \in \mathbb{R}$ and let $\Theta_{1} \sim K(\theta,\cdot)$. We consider three cases: $\theta < - s \sigma$, $\theta > s \sigma$, or $\theta \in [- s \sigma, s \sigma]$.

In the first case,
\begin{multline*}
\frac{2}{s \sigma} \mathbb{E}[V(\Theta_{1})] \\
= \int_{-s \sigma}^{0}V(\theta + x)a_{n}(\theta,\theta+x) dx + \int_{-s \sigma}^{0}V(\theta)(1-a_{n}(\theta,\theta+x)) dx + \int_{0}^{s \sigma} V(\theta + x)dx \\
=V(\theta) \left( \int_{-s \sigma}^{0} (V(x)-1) a_{n}(\theta,\theta+x) dx + \frac{1}{s \sigma} + \int_{0}^{s \sigma} V(x) dx \right) \\
= V(\theta) \left( \int_{-s \sigma}^{0} (V(x)-1) a_{n}(\theta,\theta+x) dx + \frac{1}{s \sigma} + \frac{\sigma}{\alpha}(1 - e^{-\alpha s}) \right). \\
\end{multline*}
For fixed $s$ and all $0 < \sigma < \alpha < 0.1$, we can bound the last term and continue:
\begin{multline*}
\frac{2}{s \sigma} \mathbb{E}[V(\Theta_{1})] \\
\leq V(\theta) \left( \frac{1.1}{s \sigma} + \int_{-s \sigma}^{0} (V(x)-1) a_{n}(\theta,\theta+x) dx  \right) \\ 
= V(\theta) \left( \frac{1.1}{s \sigma} + \int_{-s \sigma}^{0} (e^{\frac{\alpha}{\sigma}|x|}-1) e^{\frac{1}{2 s c^{2} \sigma^{2}}(\theta^{2} - (\theta + x)^{2}) } dx  \right) \\
= V(\theta) \left( \frac{1.1}{s \sigma} + \int_{-s \sigma}^{0} (e^{\frac{\alpha}{\sigma}|x|}-1) e^{-\frac{1}{2 s c^{2} \sigma^{2}}(2 \theta x + x^{2}) } dx  \right) \\
\leq V(\theta) \left( \frac{1.1}{s \sigma} + \int_{-s \sigma}^{0} e^{-\frac{\alpha}{\sigma}x-\frac{1}{2 s c^{2} \sigma^{2}}(2 \theta x + x^{2}) } dx  \right).
\end{multline*}

Investigating the two terms appearing inside the exponential, we note that for small $\alpha$, the second term dominates. Thus, for sufficiently small $0 < \sigma < \alpha < \epsilon$, we have 
\[ 
\frac{2}{s \sigma} \mathbb{E}[V(\Theta_{1})] \leq V(\theta)  \frac{1.2}{s \sigma}.
\] 
This proves the desired inequality in the first case.

In the second case, we simply note that $|\Theta_{1}| \leq |\theta| + s \sigma$, so 
\[ 
V(\Theta_{1}) \leq V(2 s \sigma) = e^{2 \alpha s} \leq 0.7 \, V(\theta) + e^{2 \alpha s}.
\] 
This proves the desired inequality in the second case.

The third case is essentially identical to the first case.
\end{proof}

\vfill

\end{document}